\newif\iflncs
\lncsfalse
\newif\ifnotes
\notestrue

\iflncs
\documentclass[orivec,runningheads,envcountsect,envcountsame,a4paper]{llncs}
\else
\documentclass[11pt]{article}
\fi

\iflncs
\usepackage{breakcites}
\usepackage[T1]{fontenc}
\else
\usepackage{amsthm}
\usepackage{fullpage}
\fi
\usepackage{subcaption}
\usepackage{algorithm}
\usepackage[noend]{algpseudocode}
\usepackage{setspace}
\usepackage{amsfonts}
\usepackage{amsmath}
\usepackage{bbm}
\usepackage{complexity}
\usepackage{mdframed}
\usepackage{framed}
\usepackage{float}
\usepackage{array}
\usepackage{enumitem}
\usepackage[dvipsnames]{xcolor} 
\usepackage{graphicx}
\usepackage{tikz}
\usepackage{multirow}
\usepackage{pdfpages}
\usepackage{mathtools}

\iflncs
\usepackage[colorlinks=true,linkcolor=blue,citecolor=blue]{hyperref}
\else
\usepackage[colorlinks=true,linkcolor=Maroon,citecolor=Maroon]{hyperref}
\fi
\usepackage{thm-restate}
\usepackage{orcidlink}                    
\usepackage{algpseudocode} 
\usepackage[capitalize]{cleveref}  

\newcommand{\remove}[1]{}

\ifnotes
\newcommand{\knote}[1]{[{\footnotesize \textsf{\color{red}{\bf Kel Zin:} { {#1}}}}]}
\definecolor{ao}{rgb}{0.0, 0.5, 0.0}
\newcommand{\pnote}[1]{[{\footnotesize \textsf{\color{cyan}{\bf Prashant:} { {#1}}}}]}
\newcommand{\dnote}[1]{[{\footnotesize \textsf{\color{blue}{\bf Divesh:} { {#1}}}}]}
\newcommand{\hnote}[1]{[{\footnotesize \textsf{\color{orange}{\bf Hai:} { {#1}}}}]}
\newcommand{\rnote}[1]{[{\footnotesize \textsf{\color{purple}{\bf Rishav:} { {#1}}}}]}
\else
\newcommand{\knote}[1]{}
\definecolor{ao}{rgb}{0.0, 0.5, 0.0}
\newcommand{\pnote}[1]{}
\newcommand{\dnote}[1]{}
\newcommand{\hnote}[1]{}
\newcommand{\rnote}[1]{}
\fi
\definecolor{mygreen}{RGB}{0,128,0} 
\newcommand{\resolved}[1]{{\color{mygreen}{[Resolved]}}}

\iflncs
\input{lncs_issues}

\spnewtheorem{claim}[theorem]{Claim}{\bfseries}{\itshape}
\else

\newtheorem{theorem}{Theorem}[section]

\newtheorem{claim}[theorem]{Claim}
\newtheorem{corollary}[theorem]{Corollary}

\newtheorem{definition}[theorem]{Definition}
\theoremstyle{definition}
\newtheorem{remark}[theorem]{Remark}
\fi

\iflncs
\renewcommand{\paragraph}[1]{\;\newline \noindent \textbf{#1}}
\fi

\newenvironment{proofof}[1]{\begin{proof}[\textit{Proof of #1}]}{\end{proof}}

\crefname{definition}{Definition}{Definitions}
\crefname{sub-definition}{Definition}{Definitions}
\crefname{example}{Example}{Examples}
\crefname{exercise}{Exercise}{Exercises}
\crefname{property}{Property}{Properties}
\crefname{question}{Question}{Questions}
\crefname{solution}{Solution}{Solutions}
\crefname{theorem}{Theorem}{Theorems}
\crefname{informaltheorem}{Theorem}{Theorems}
\crefname{proposition}{Proposition}{Propositions}
\crefname{problem}{Problem}{Problems}
\crefname{lemma}{Lemma}{Lemmas}
\crefname{conjecture}{Conjecture}{Conjectures}
\crefname{corollary}{Corollary}{Corollaries}
\crefname{fact}{Fact}{Facts}
\crefname{claim}{Claim}{Claims}
\crefname{remark}{Remark}{Remarks}
\crefname{note}{Note}{Notes}
\crefname{figure}{Figure}{Figure}
\crefname{case}{Case}{Cases}
\crefname{proofsketch}{Proof Sketch}{Proof Sketches}
\RequirePackage{algorithm} 
\RequirePackage{algpseudocode}

\makeatletter
\newcommand{\linelabel}[1]{%
  \begingroup
  \def\@currentcounter{ALG@line}%
  \label{#1}%
  \endgroup
}
\makeatother
\crefname{ALG@line}{line}{lines}
\Crefname{ALG@line}{Line}{Lines}

\def\ddefloop#1{\ifx\ddefloop#1\else\ddef{#1}\expandafter\ddefloop\fi}

\def\ddef#1{\expandafter\def\csname bb#1\endcsname{\ensuremath{\mathbb{#1}}}}
\ddefloop ABCDEFGHIJKLMNOPQRSTUVWXYZ\ddefloop

\def\ddef#1{\expandafter\def\csname c#1\endcsname{\ensuremath{\mathcal{#1}}}}
\ddefloop ABCDEFGHIJKLMNOPQRSTUVWXYZ\ddefloop

\def\ddef#1{\expandafter\def\csname v#1\endcsname{\ensuremath{\overline{#1}}}}
\ddefloop ABCDEFGHIJKLMNOPQRSTUVWXYZabcdefghijklmnopqrstuvwxyz\ddefloop

\newcommand{\algfont}[1]{\mathsf{#1}}
\def\ddef#1{\expandafter\def\csname alg#1\endcsname{\ensuremath{\algfont{#1}}}}
\ddefloop ABCDEFGHIJKLMNOPQRSTUVWXYZ\ddefloop

\newcommand{\langfont}[1]{\mathnormal{#1}}
\def\ddef#1{\expandafter\def\csname lang#1\endcsname{\ensuremath{\langfont{#1}}}}
\ddefloop ABCDEFGHIJKLMNOPQRSTUVWXYZ\ddefloop

\def\ddef#1{\expandafter\def\csname v#1\endcsname{\ensuremath{\boldsymbol{\csname #1\endcsname}}}}
\ddefloop {alpha}{beta}{gamma}{delta}{epsilon}{varepsilon}{zeta}{eta}{theta}{vartheta}{iota}{kappa}{lambda}{mu}{nu}{xi}{pi}{varpi}{rho}{varrho}{sigma}{varsigma}{tau}{upsilon}{phi}{varphi}{chi}{psi}{omega}{Gamma}{Delta}{Theta}{Lambda}{Xi}{Pi}{Sigma}{varSigma}{Upsilon}{Phi}{Psi}{Omega}{ell}\ddefloop

\newcommand{\eps}{\varepsilon}

\newcommand{\Field}[1]{\mathbb{F}_{#1}}
\newcommand{\F}{\mathbb{F}}

\DeclareMathOperator*{\Exp} {{\mathbf{E}}}

\newcommand{\pr}[1]         {\Pr\left[ #1 \right]}

\let\exponential\exp
\renewcommand{\exp}[1]      {\Exp\left[ #1 \right]}

\renewcommand{\poly}{\mathrm{poly}}

\DeclareMathOperator{\supp} {Supp}

\newcommand{\abs}[1]        {\left| #1\right|}

\newcommand{\ceil}[1]       {\left\lceil #1 \right\rceil}

\usepackage{listofitems}
\newcommand\cycle[2][\,]{%
  \readlist\thecycle{#2}%
  (\foreachitem\i\in\thecycle{\ifnum\icnt=1\else#1\fi\i})%
}

\newcommand{\Dsd}{\ensuremath{\algD_{\text{SD}}}}
\newcommand{\Ssd}{\ensuremath{\algS_{\text{SD}}}}
\newcommand{\LPN}{\mathsf{LPN}}
\newcommand{\Dlpn}{\ensuremath{\algD_{\text{LPN}}}}
\newcommand{\Slpn}{\ensuremath{\algS_{\text{LPN}}}}

\newcommand{\Ber}{ \mathsf{Ber}}

\newcommand{\wt}{\mathrm{wt}}

\begin{document}

\title{Towards Worst-case Hardness for Low-Noise LPN}
\author{Divesh Aggarwal\thanks{\orcidlink{0000-0002-3841-0262} Email: \href{mailto:divesh@comp.nus.edu.sg}{divesh@comp.nus.edu.sg}} 
\and Rishav Gupta\thanks{\orcidlink{0009-0002-4424-5432} Email: \href{mailto:rishavg@u.nus.edu}{rishavg@u.nus.edu}}
\and Hai Hoang Nguyen\thanks{\orcidlink{0000-0003-2371-6231} Email: \href{mailto:hai.h.nguyen@nus.edu.sg}{hai.h.nguyen@nus.edu.sg}}
\and Kel Zin Tan\thanks{\orcidlink{0009-0008-1099-0892} Email: \href{mailto:kelzin@u.nus.edu}{kelzin@u.nus.edu}} 
\and Prashant Nalini Vasudevan\thanks{\orcidlink{0000-0001-6880-795X} Email: \href{mailto:prashvas@nus.edu.sg}{prashvas@nus.edu.sg}}
\\\\ 
National University of Singapore
}
\date{\today}

\pagenumbering{roman}

\maketitle

\begin{abstract}

The hardness of the Learning Parity with Noise (LPN) problem is a foundational assumption in cryptography, forming the basis of constructions ranging from symmetric-key primitives to public-key encryption and beyond. A central open question is whether the average-case hardness of LPN can be based on worst-case complexity assumptions, as has been achieved for the analogous Learning With Errors (LWE) problem.

Existing worst-case-to-average-case reductions for LPN~\cite{EC:BLVW19, C:YuZha21} rely on statistical smoothing of linear codes, which inherently limits the resulting average-case hardness to noise rates as large as $1/2 - 1/\mathrm{poly}(n)$, which is insufficient for public-key applications.

We explore a new approach towards obtaining such reductions: rather than requiring that random sparse combinations of the rows of the generator matrix of a code be \emph{statistically} close to uniform, we only require that they be \emph{computationally indistinguishable} from uniform. This leads to a clean win-win structure: we show that any efficient LPN solver can be transformed into a pair of efficient algorithms $(S, D)$ such that for \emph{every} matrix $A$ of appropriate dimensions over $\F_2$, either $S$ decodes the code generated by $A$ from random noise, or $D$ distinguishes random noisy codewords of the dual of this code from uniform.

By instantiating this reduction with appropriate parameters, we obtain the average-case hardness of LPN with inverse-polynomial noise rate~$n^{-\alpha}$ for any constant $\alpha < 1$, assuming the worst-case simultaneous hardness of decoding a code from random noise and distinguishing random noisy codewords of its dual from uniform. In particular, setting $\alpha = 1/2$, our reduction yields LPN hardness in the parameter regime required for Alekhnovich's construction of public-key encryption~\cite{FOCS:Alekhnovich03}, a regime that was previously inaccessible via worst-case reductions.



\end{abstract}


\newpage
\tableofcontents
\newpage

\pagenumbering{arabic}

\section{Introduction}
\label{sec:intro}

Learning Parity with Noise (LPN)~\cite{C:BFKL93} is essentially the problem of decoding random noisy codewords of a random linear code. An instance consists of a uniformly random matrix $A \gets \F_2^{m\times n}$ and a vector $A\cdot s + e \in \F_2^m$, where $s\gets\F_2^n$ is a uniformly random vector and $e\gets Ber(\eta)^m$ is a vector of Bernoulli variables, where $\eta\in(0,1/2)$ is called the \emph{noise rate}.
In the Search LPN problem, the task is to recover $s$ given $(A,As+e)$. In the Decision LPN problem, the task is to distinguish this distribution from $(A,b)$, where $b$ is an independent uniformly random vector in $\F_2^m$. These variants are known to be equivalent~\cite{C:AppIshKus07,C:MicMol11}.

The LPN problem has been of significant consequence to cryptography since its first consideration by Blum et al.~\cite{C:BFKL93}. It is one of the few known sources of secure Public-Key Encryption (PKE)~\cite{FOCS:Alekhnovich03, C:YuZha16}, and has been used to construct various advanced cryptographic primitives (see, for instance, the survey by Pietrzak~\cite{Pie12}). It may be interpreted as the problem of solving a system of noisy linear equations, and for large enough $\eta$, worst-case versions of this problem are known to be $\NP$-hard~\cite{AK14}. The best algorithms we have for LPN run in time $2^{O(\eta\cdot n)}$ for any non-trivial $m$~\cite{Prange62,Stern88,Dumer91,AC:MayMeuTho11,EC:BJMM12,EC:MayOze15,WCC:BotMay17,PQCRYPTO:BotMay18,EC:DEEK24}, or in time $2^{O(n/\log{n})}$ if $m > 2^{\Omega(n/\log{n})}$~\cite{JACM:BluKalWas03}, or $2^{O(n/\log\log{n})}$ if $m > n^{1+\Omega(1)}$~\cite{RANDOM:Lyubashevsky05}.

In many ways, LPN is similar to its analogue in the Euclidean metric, the Learning With Errors (LWE) problem~\cite{JACM:Regev09}. One notable difference, however, is that the (average-case) hardness of the LWE problem is supported by the worst-case hardness of lattice problems~\cite{JACM:Regev09,STOC:Peikert09}, leading to a construction of PKE and other primitives from the worst-case hardness of these problems. With LPN, the worst-case to average-case reductions we have are much more limited.

\paragraph{Known Reductions for LPN.} The only known method for reducing from the plausible worst-case hardness of any problem to LPN is that of Brakerski et al.~\cite{EC:BLVW19}. Their reduction is from a worst-case version of the LPN problem itself, called the Nearest Codeword Problem (NCP). Here, given a matrix $A\in\F_2^{m\times n}$ and vector $As+e\in\F_2^m$ for some $e$ guaranteed to have Hamming weight at most $\eta\cdot m$, the problem is to recover $s$. They show that the worst-case hardness of NCP (with some minor restrictions) with noise rate $\eta = O(\log^2{n}/n)$ implies the average-case hardness of LPN with noise rate $(1/2-1/\poly(n))$.

With additional careful analysis, Yu and Zhang~\cite{C:YuZha21} extend this reduction to sub-exponential-time algorithms. For example, they show that if NCP with noise rate $\eta = O(n^{-1/2})$ and $m = 2^{\Omega(n^{1/2})}$ is worst-case hard for algorithms running in time $2^{O(n^{1/2})}$, then LPN with constant noise-rate and similar values of $m$ is also hard for such algorithms. They also point out that the hardness of LWE with certain noise parameters implies hardness of an analogue of LPN over large fields with noise rate close to $1$.

Apart from worst-case-to-average-case reductions, hardness amplification results are also known for LPN~\cite{AGL26}. Roughly speaking, these results show that an algorithm solving LPN with secret length $k n$ and noise rate $k\eta$ with success probability $\varepsilon$ can be used to solve LPN with secret length $n$ and noise rate $\eta$ with success probability at least $1-\delta$, where $k=\Theta\!\left(\frac{1}{\delta}\log\frac{1}{\varepsilon}\right)$.

\paragraph{Limits of Known Reductions.} The aforementioned results are already quite remarkable as they imply the possibility of symmetric-key cryptography based on worst-case hardness assumptions. Unfortunately, they fall well short of providing the hardness of LPN required for public-key applications, which is hardness at noise rate $O(n^{-1/2})$~\cite{FOCS:Alekhnovich03} against polynomial-time algorithms, or at constant noise rate against $2^{\omega(n^{1/2})}$-time algorithms~\cite{C:YuZha16}.

Further, the worst-case hardness assumptions in their hypotheses are quite strong. Recall that \cite{EC:BLVW19} requires hardness of NCP with a noise rate of $O(\log^2{n}/n)$. This is only slightly larger than $O(\log{n}/n)$, at which rate NCP can be solved in polynomial time~\cite{Prange62}.
Further, if arbitrary pre-processing of the matrix $A$ is allowed, even NCP with noise rate $O(\log^2{n}/n)$ can be solved in polynomial time~\cite{ITCS:BCLV26}.

\medskip
The source of these limitations is as follows. The core of their reduction is, given an arbitrary matrix $A$, to consider the vector $r^T A$, where $r\in\F_2^m$ is a random vector with Hamming weight $w$. Clearly, if $w > n$ and $A$ is sufficiently non-degenerate, then $r^TA$ will be close to being uniformly random. They show that this is true even if $w = O(n/\log{n})$ and $m$ is a large enough polynomial in $n$. Then, given an NCP instance $(A,As+e)$, they generate the rows of their LPN instance as $(r^TA, r^T(As+e))$. The first part $r^T A$ is uniform as required, and the second part is $(r^TA)s + r^Te$, with $r^Te$ playing the part of the noise. This phenomenon of the distribution of $r^TA$ being close to uniform for an appropriate distribution of $r$ is referred to as \emph{smoothing},\footnote{Strictly speaking, ``smoothing'' actually refers to a dual phenomenon that this is known to be equivalent to~\cite{DCC:PatBar25}.} and has been studied in detail for this and other distributions of $r$, and in various metrics~\cite{EC:BLVW19,C:YuZha21,ITIT:DDRT23,PKC:DebRes25,DCC:PatBar25}.

This method of generating the instance causes the noise rate to increase. If $e$ had noise rate $\eta$, then $r^Te$ will have noise rate $\alpha \approx 1/2 - e^{-2\eta t}$. This $\alpha$ has to be non-negligibly removed from $1/2$ for the problem to be meaningful. So if $w$ is set to $\Omega(n/\log{n})$, then only noise rates $\eta$ between $\omega(\log{n}/n)$ and $O(\log^2n/n)$ are useful, and the resulting $\alpha$ will never be smaller than $1/2-o(1)$.

\medskip
One way to try to deal with this issue is to set $w$ to be something smaller. However, it is easy to check that if $w = o(n/\log{n})$, then for any $m$ that is polynomial in $n$, a random vector of Hamming weight $w$ does not have enough entropy for the distribution of $r^TA\in\F_2^n$ to be close to uniform. This remains true for other natural distributions such as Bernoulli vectors with expected Hamming weight $w$.

Another possibility is that instead of using the same distribution for every matrix $A$, one could try to use a distribution of $r$ specifically tailored to the given $A$ such that $r^TA$ will be uniform. However, without bounds on the Hamming weight of $r$, it is unclear that $r^Te$ will remain unbiased for all possible noise vectors $e$. In fact, it was recently shown by Pathegama and Barg~\cite{DCC:PatBar25} that, at least if $m = O(n)$ and the noise rate is a constant, there are matrices $A\in\F_2^{m\times n}$ for which there is no such distribution of $r$ that will satisfy both conditions: $r^TA$ is sufficiently close to uniform, and $r^Te$ has noise rate $1/2 - 1/\poly(n)$. Though not explicitly stated, their result also applies when $e$ is a random Bernoulli vector. This still leaves open the possibility of this approach working if $m$ is much larger, or if we start with hardness of NCP with a smaller noise rate, but pursuing it will require significantly new techniques.



\subsection{Our Results}
\label{sec:results}

Our proposal to improve on existing reductions is somewhat different, and is to weaken the notion of smoothing used -- instead of asking that $r^TA$ be statistically close to uniform, we only ask that it be \emph{computationally indistinguishable} from uniform. Interestingly, this again relates to decoding a linear code -- if we view $A^T$ as the parity-check matrix of its dual code, then $A^Tr$ may be regarded as the syndrome of $r$. If $B$ generates the dual code of $A$, then multiplying a noisy dual codeword $Bs+r$ by $A^T$ results in the codeword part being cancelled out, giving us: $A^T (Bs+r) = (A^TB)s + A^Tr = A^Tr$. So recovering $r$ given $r^TA$ is equivalent to decoding the dual code of $A$ with $r$ as the noise vector, and distinguishing $r^TA$ from uniform corresponds to solving a decision version of the decoding problem for the dual code of $A$. 

This now leads to a win-win situation. If $r^TA$ is computationally indistinguishable from uniform, then the existing smoothing-based reduction approach can be used to generate \emph{pseudorandom} LPN samples of the form $(r^TA,r^T(As+e))$, which can then be given to an LPN solver. Since this distribution is pseudorandom, the LPN solver will continue to work and find the secret $s$. On the other hand, if $r^TA$ is \emph{not} pseudorandom, then the distinguisher that breaks pseudorandomness is essentially solving a decision version of the decoding problem for the dual code of $A$.

\medskip
Following this observation, we show that the hardness of LPN can be derived from the worst-case simultaneous hardness of decoding both a code and its dual from random noise. An informal version of this reduction is described below. We refer to each row of $(A,As+e)$ as a \emph{sample} of LPN, and each bit of the secret $s$ as a \emph{variable}.

\begin{theorem}[Informal, see \cref{thm:main} and \cref{cor:lpnmain}]
  \label{infthm:main}
  Suppose there is an efficient algorithm that has non-negligible advantage in solving LPN with $n$ variables, $t$ samples, and noise rate $n^{-\alpha}$ for some constant $\alpha < 1$. Then, for any constant $\beta > 0$ such that $\alpha + \beta < 1$, and $m = O(t^2n^{2\beta})$, there are efficient algorithms $\algS$ and $\algD$ such that, for any $A\in\F_2^{m\times n}$, either $\algS$ can decode the code generated by $A$ from random noise of rate $n^{-(\alpha+\beta)}$, or $\algD$ can distinguish noisy codewords of the dual of this code from random for noise rate $n^\beta/m$.
\end{theorem}

To the best of our knowledge, the complexity of decoding a generic linear code and decoding its dual are not known to be negatively correlated. So we expect the complexity of the task specified in the theorem to be the minimum of the worst-case complexity of an algorithm that can decode any code generated by an $m\times n$ matrix from random noise of rate $n^{-(\alpha+\beta)}$ and that of an algorithm that can decode any code generated by an $m\times (m-n)$ matrix from random noise of rate $n^\beta/m$. As long as $\alpha$ and $\beta$ are constants such that $\alpha + \beta < 1$, the best known algorithms for both of these tasks run in sub-exponential time.

\paragraph{Public-Key Cryptography.} The most remarkable consequence of our reduction is that, making reasonable assumptions about the hardness of decoding worst-case codes and their duals from random noise, we can get the hardness of LPN required for constructions of Public-Key Encryption using the Alekhnovich paradigm~\cite{FOCS:Alekhnovich03}. In particular, this requires hardness for LPN with $t = 2n$ samples and noise rate $n^{-1/2}$, which can be obtained by setting the parameters in the theorem to, for instance, $\alpha = 0.5$, $\beta = 0.4$, and $m = O(n^{2.4})$.

\paragraph{Challenges.} While the observations stated earlier that lead to our reduction are relatively straightforward, realising them fully turns out to be quite challenging. The primary obstacle here is that computational randomness is much harder to obtain and use than actual randomness. One detail that we have been glossing over in our discussion so far is that for the smoothing reduction approach to work, it is not sufficient to show that $r^TA$ is close to uniform -- one also needs to show that the joint distribution of $(r^TA,r^Te)$ is close to $(U_n,Ber(\alpha))$, where $U_n$ is the uniform distribution over $n$ bits and $\alpha$ is the noise rate of LPN one is reducing to. In the statistical setting of \cite{EC:BLVW19,C:YuZha21}, the Fourier analysis they use can easily show this using nearly the same argument that shows that $r^TA$ is close to uniform.

In our computational setting, however, we do not have powerful tools like Fourier analysis that can readily incorporate this additional biased bit. So even if we are promised that $r^TA$ is pseudorandom, and we also know that $r^Te$ is a Bernoulli variable with the appropriate parameter $\alpha$, it is non-trivial to argue that $(r^TA,r^Te)$ is computationally indistinguishable from $(U_n,Ber(\alpha))$. This is because the leakage $r^Te$ could lead to $r^TA$ becoming slightly distinguishable from uniform, and this could add up over many such samples to render the resulting distribution easily distinguishable from an LPN instance.

It is due to this issue that our reduction only produces an $\algS$ that decodes from random noise, whereas the earlier smoothing reductions could produce algorithms that can decode from worst-case noise. We show that for most choices of the set of vectors $r$ chosen while constructing the LPN instance, even after fixing the $r$'s, the randomness in the $e$ ensures that the set of bits $r^Te$ are independent Bernoulli variables. This also requires us to use a different error model for the dual code -- the $r$'s are chosen to be vectors of some fixed Hamming weight, rather than as Bernoulli vectors (see \cref{thm:main}). 

Nevertheless, we believe that our reduction should work even with $r$ being a Bernoulli vector and $e$ being any worst-case error vector of bounded Hamming weight. Proving that it works seems quite challenging, however, and will likely require the development of new analytical tools for computational entropy.

\subsection{Future Directions}

Our work opens up several natural directions for further investigation.

\medskip\noindent\textbf{More standard worst-case foundations.}
The worst-case assumption that we need, the simultaneous hardness of
decoding a code and distinguishing syndromes of its dual from uniform, is non-standard and, to
our knowledge, has not been independently studied. An important next step is to
determine whether our approach can be instantiated from
more widely studied worst-case assumptions, such as the hardness of the Nearest
Codeword Problem (NCP) alone, or classical problems like the Minimum Distance
Problem (MDP). More broadly, one would like to understand whether the joint
hardness of a code and its dual is inherently needed to improve the smoothing approach to reductions, or whether there is some way around it.

\medskip\noindent\textbf{Worst-case noise.}
Our reduction produces a decoder~$S$ that works against \emph{random} Bernoulli
noise, whereas the smoothing-based reductions of Brakerski et
al.~\cite{EC:BLVW19} and Yu--Zhang~\cite{C:YuZha21} yield decoders for \emph{worst-case}
bounded-weight noise. Extending our approach to worst-case noise would
significantly strengthen the result, but appears to require new tools for
reasoning about computational entropy in the presence of adversarially chosen
error vectors. We believe this is a tractable and important challenge.

\subsection{Paper Outline}
In \cref{sec:prelims}, we recall the notation and technical tools used throughout the paper.
In \cref{sec:reduction}, we present our main reduction. 
In \cref{sec:corollaries}, we instantiate our reduction with suitable parameters to derive useful hardness regimes for LPN.








\section{Preliminaries}
\label{sec:prelims}



\paragraph{Notation} 
We use $x_i$ to denote the $i$-th bit of the binary string (vector) $x$ and $A[i]$ as the $i$-th row of the matrix $A$.
Denote $D_1 \equiv D_2$ if the distributions $D_1$ and $D_2$ are the same.
The notation $x \gets D$ means that $x$ is sampled from the distribution $D$; when $D$ is a set (we abuse notation), it means $x$ is sampled uniformly from the set.
Let $\Field{p} $ denote the order-$p$ prime field. 
$\mathrm{Ber}(\eta)$ is the Bernoulli distribution with parameter $\eta$. 
Write $[n] \coloneqq \{1,2,\dotsc, n\}$.
The notation $\binom{[n]}{s}$ denotes the set of all vectors in $\F_2^n $ with Hamming weight $s$.
Next, $\operatorname{wt}({v})$ denotes the Hamming weight of $v \in \F_2^n$, and $\supp(v) \coloneqq \{ i \in [n]\ |\ v_i \neq 0\}$ is the support of the vector.






\begin{restatable}{fact}{pilinguplemma}\emph{(Piling Up Lemma).}
    \label{fact:sum_bern}
    Let $X_1, X_2, \dots,X_k$ be independent random variables where $X_i \gets\Ber(\eta)$, then the random variable $S = X_1 + X_2 + \dots + X_k$, where addition is over $\Field{2}$, follows the distribution $\Ber\left(\dfrac{1-(1-2\eta)^k}{2}\right)$.
\end{restatable}

\begin{proof}
    Deferred to \cref{proof:piling_up_lemma}
\end{proof}


We now define the set of vectors with fixed Hamming weight.
\begin{definition}[Fixed weight vector]
    Denote $\mathcal{W}(n,w)$ as the set of vector $v \in \Field{2}^n$ with $\wt(v) = w$.
\end{definition}
Next, we formally define the Learning Parity with Noise problem $\mathsf{(LPN)}$.
\begin{definition}[Learning Parity with Noise (LPN)]
    Let $n,m \in \mathbb{Z}$, and $\eta \in [0, 1]$. For sample count $m$, secret dimension $n$, and noise parameter $\eta$, algorithms $\Slpn$ and $\Dlpn$ are said to have advantage $\eps \in [0, 1]$ in solving the \emph{Search LPN} and \emph{Decision LPN} problems, respectively, if they satisfy the following conditions:
    \begin{itemize}[itemsep=0pt]
        \item \emph{Search LPN:} $\Slpn$ recovers the secret $s$ with probability at least $\eps$:
        \[
            \Pr_{\substack{A \gets \Field{2}^{m \times n} \\ s \gets \Field{2}^n \\ e \gets \mathrm{Ber}(\eta)^m}}[\Slpn(A, As + e) = s] \geq \eps
       . \]
        
        \item \emph{Decision LPN:} $\Dlpn$ distinguishes LPN samples from uniform with advantage at least $\eps$:
        \[
            \abs{\Pr_{\substack{A \gets \Field{2}^{m \times n} \\ s \gets \Field{2}^n \\ e \gets \mathrm{Ber}(\eta)^m}}[\Dlpn(A, As + e) = 1] - \Pr_{\substack{A \gets \Field{2}^{m \times n} \\ b \gets \Field{2}^m}}[\Dlpn(A, b) = 1]} \geq \eps.
        \]
    \end{itemize}
\end{definition}
We now define the Fixed Weight Syndrome Decoding problem $\mathsf{SD}$, which can be viewed as a slight variation of the dual form of $\mathsf{LPN}$.
\begin{definition}[Fixed Weight Random Syndrome Decoding (SD)]
    Let $n,m \in \mathbb{Z}$, and $\theta \in [0, 1]$. For code length $m$, code dimension $n$, and noise parameter $\theta$, algorithms $\Ssd$ and $\Dsd$ are said to have advantage $\eps \in [0, 1]$ in solving the \emph{Search SD} and \emph{Decision SD} problems, respectively, if they satisfy the following conditions:
    \begin{itemize}[itemsep=0pt]
        \item \emph{Search SD:} $\Ssd$ decodes a syndrome with probability at least $\eps$:
        \[
            \Pr_{\substack{A \gets \Field{2}^{m \times n}\\ r \gets \mathcal{W}(m, \theta \cdot m)}}[\Ssd(A, r^TA) = r] \geq \eps.
        \]
        
        \item \emph{Decision SD:} $\Dsd$ distinguishes a syndrome from uniform with advantage at least $\eps$:
        \[
            \abs{\Pr_{\substack{A \gets \Field{2}^{m \times n}\\ r \gets \mathcal{W}(m, \theta \cdot m)}}[\Dsd(A, r^TA) = 1] - \Pr_{\substack{A \gets \Field{2}^{m \times n} \\ c \gets \Field{2}^n}}[\Dsd(A, c^T) = 1]} \geq \eps.
        \]
    \end{itemize}
\end{definition}

\begin{remark}
    Our definition of Syndrome Decoding Problem deviates from the more natural dual to LPN in that the distribution of $r$ is a randomly sampled fixed weight vector rather than a Bernoulli vector.
\end{remark}

\section{Average-Case Reduction}
\label{sec:reduction}

Our main contribution is the following theorem, which shows that for every worst-case code generated by a matrix $A$, an $\LPN$ solver yields one of two consequences: either it can be used to solve a random decoding instance associated with $A$, or it can be used to distinguish fixed-weight random syndrome decoding instance with $A^T$ as the parity-check matrix from a random string.

\begin{theorem}
  \label{thm:main}
  Consider $n, m, t \in \mathbb{Z}$ and $\eta, \theta \in (0,1)$ such that $t^2\theta^2m < 1$. Suppose there exists an algorithm for Search LPN with sample count $t$, secret dimension $n$, and noise parameter $\tau =  (1 - (1-2\eta)^{\theta \cdot m})/2$, which runs in time $T$ and has advantage $\eps$. Then there exist algorithms $\algD$ and $\algS$ such that for every matrix $A \in \Field{2}^{m \times n}$, at least one of the following is true:
  \begin{itemize}[itemsep=0pt]
    \item $\algS$ runs in  time $\poly\left(n, m, T, (1/\eps)^{\log_m(m/\eps)}\right)$ and decodes codewords of $A$ from random noise such that, \[\Pr_{\substack{s \gets \Field{2}^n\\ e \gets Ber(\eta)^m}}[\algS(A, As+e) = s^\prime\  \mathrm{and}\  \wt(As^\prime - (As + e)) \leq 2\eta m] \geq \eps/2 - o(\eps).
    \]
    \item $\algD$ runs in time $\poly(n,m) + T$ and distinguishes a sparse combination of rows of $A$ from a uniform string such that,
    \[
    \Pr_{c \gets \Field{2}^n}[\algD(A, c^T) = 1] - \Pr_{r \gets {\cW(m,\theta \cdot m)}}[\algD(A, r^TA) = 1]  \geq \eps/2t.
    \]
  \end{itemize}
\end{theorem}

The rest of this section constitutes the proof of \cref{thm:main}. Fix any set of parameters $n$, $m$, $t$, $\eta$, $\theta$, and $\tau$ that satisfy the conditions in the theorem statement. Suppose there is an algorithm $\Slpn$ for Search LPN that has running time $T$ and advantage $\eps$ as specified in the theorem. The main idea in the reduction is to check how good $\Slpn$ is at decoding a code whose generator matrix is derived by randomly combining rows of a given matrix $A\in\F_2^{m\times n}$, and based on this, using it to either decode the code generated by $A$, or the code whose parity check matrix is $A$.

For any $A\in\F_2^{m\times n}$, let $Adv(A)$ be the advantage of $\Slpn$ at decoding a noisy codeword of the code generated by the matrix $R\cdot A \in \Field{2}^{t \times n}$, where $R$ is a $t \times m$ matrix and each row of $R$ is an independently sampled random vector of Hamming weight $\theta \cdot m$, written as $R \gets \cW(m, \theta \cdot m)^{t}$:
\begin{equation}
  \label{eq:advantage}
  Adv(A) = 
  \Pr_{\substack{s \gets \Field{2}^{n}\\ e \gets \mathrm{Ber}(\tau)^t\\ R \gets \cW(m, \theta \cdot m)^{t}}}[\Slpn(RA, RAs + e) = s]
\end{equation}
We can then split into two cases. If $Adv(A)$ is sufficiently large (i.e. $Adv(A) \geq \eps/2$), then $\Slpn$ can be used to construct an algorithm that decodes a noisy codeword $As+e$ of $A$ by running $\Slpn$ on $(RA,R(As+e))$. On the other hand, if $Adv(A)$ is sufficiently smaller than $\eps$ (i.e. $Adv(A) < \eps/2$), then this means $RA$ does not look random to $\Slpn$, and we can use a hybrid argument to distinguish the syndrome $r^TA$ from uniform when $r\gets\cW(m,\theta\cdot m)$. 

\medskip
We next state the lemmas capturing the properties of the two resulting algorithms. Note that these lemmas are stated under the assumption that the algorithm $\Slpn$ as described above exists. The algorithms themselves and the proofs of these lemmas are presented in \cref{sec:adv_large,sec:adv_small}. Together, these lemmas directly imply \cref{thm:main}. 

\begin{restatable}[Large Advantage]{lemma}{lemmaadvlarge}
  \label{lem:adv_large}
  There exists an algorithm $\algS$ that, for any $A\in\F_2^{m\times n}$ such that $Adv(A) \geq \eps/2$, decodes the code generated by $A$ from random Bernoulli noise of rate $\eta$, with success probability at least $\eps/2-o(\eps)$. More precisely,
  \[
  \Pr_{\substack{s \gets \Field{2}^n\\ e \gets Ber(\eta)^m}}[\algS(A, As+e) = s^\prime\  \mathrm{and}\  \wt(As^\prime - (As + e)) \leq 2\eta m] \geq \eps/2 - o(\eps).
  \]
  Further, $\algS$ runs in time $\poly\left(n, m, T, (1/\eps)^{\log_m(m/\eps)}\right)$.
\end{restatable} 

\begin{restatable}[Small Advantage]{lemma}{lemmaadvsmall}
  \label{lem:adv_small}
  There exists an algorithm $\algD$ that, for any $A\in\F_2^{m\times n}$ such that $Adv(A) < \eps/2$, distinguishes between random syndromes of the code with parity check matrix $A$ and the uniform distribution with advantage at least $\eps/2t$. More precisely,
  \[
  \Pr_{c \gets \Field{2}^n}[\algD(A, c^T) = 1] - \Pr_{\substack{r \gets \cW(m, \theta \cdot m)}}[\algD(A, r^TA) = 1] \geq \eps/2t.
  \]
  Further, $\algD$ runs in time $T+\poly(n,m)$.
\end{restatable}

\subsection{Large Advantage}
\label{sec:adv_large}

In this section, we prove \cref{lem:adv_large}, about the ability of an algorithm $\algS$ to decode the code generated by $A$ when $Adv(A) \geq \eps/2$.

\lemmaadvlarge*


The algorithm $\algS$ that proves the lemma is described below, followed by the proof of the lemma. Here, $\Slpn$ is the assumed search algorithm for LPN that has running time $T$ and advantange $\eps$.

\begin{algorithm}[h!]
  \caption{Algorithm $\algS(A \in \Field{2}^{m \times n}, b \in \Field{2}^{m})$} \label{alg:alg1}
  \begin{spacing}{1.2}
    \begin{algorithmic}[1]
      \State Sample matrix $R \gets \cW(m, \theta \cdot m)^{t}$. \Comment{Smoothing matrix}
      \State Set $\mathrm{threshold} = \max(8\log_m(m/\eps)\log(1/\eps), 12)$. \Comment{Threshold}
      \State Set $C = \{ j \in [t] \mid \exists i < j: \supp(R[i]) \cap \supp(R[j]) \neq \emptyset \}$. 
      \Comment{Collision Set}
      \State Convert C into a list with arbitrary ordering.
      \If{$|C| \geq \mathrm{threshold}$}
      \State \Return $\perp$.\linelabel{step:abort}\Comment{Abort}
      \EndIf
      
      \For{each $v \in \Field{2}^{|C|}$}
      \State Set $b^\prime \gets Rb$.
      \For{each $j$ in $C$} \Comment{Iterate through collided rows}
      \State Replace $b^\prime_{j} \gets (v_j + \Ber(\tau))$  \linelabel{algo:replace} \Comment{Recall that $\tau =  (1 - (1-2\eta)^{\theta \cdot m})/2$}
      \EndFor
      \State Query $s^\prime \gets \Slpn(RA, b^\prime)$. \linelabel{algo:query} \Comment{Query oracle}
      \If{$\wt(b - As^\prime) \leq 2\eta m$}    \linelabel{algo:verify}\Comment{Verify answer}
      \State \Return $s^\prime$. 
      \EndIf
      \EndFor
      \State \Return $\perp$
    \end{algorithmic}
  \end{spacing}
\end{algorithm}

\begin{proofof}{\cref{lem:adv_large}}
  Fix an $A$ such that $Adv(A) \geq \eps/2$. Recall that this means that $\Slpn$ has $\eps/2$ probability of successfully recovering $s$ when given a sample from the distribution $(RA, RAs + \Ber(\tau)^t)$, with $\tau =  (1 - (1-2\eta)^{\theta \cdot m})/2$. To show the correctness of the algorithm, it suffices to prove that in at least one iteration, the input $(RA, b^\prime)$ provided to $\Slpn$ in \cref{algo:query} has the same distribution as $(RA, RAs + \Ber(\tau)^t)$, and also that the algorithm does not prematurely terminate at \cref{step:abort}.\\\\
  In the following claim, we show that the algorithm does not abort in \cref{step:abort} with high probability. 
  \begin{claim}
    The probability that $|C| \geq \mathrm{threshold} = \max(8 \log_m(m/\eps)\log(1/\eps), 12)$ is at most $o(\eps)$. Formally,
    \[\Pr\left[|C| \geq \max(8 \log_m(m/\eps)\log(1/\eps), 12)\right]
    \leq o(\eps).\]
  \end{claim}
  
  \begin{proof}
    Using union bound, we first upperbound the probability that a column in $R \gets \cW(m, \theta \cdot m)^t$ has at least $k$ bits that are 1. Let $P_{i} \in \Field{2}^{t}$ be the $i$-th column vector of $R$. For any $i \in [m]$,
    \begin{align}                
      \pr{\wt(P_i) \geq k } \leq {t \choose k }(\dfrac{\theta m}{m})^k \leq (\dfrac{et\theta}{k})^k.
    \end{align}
    Since $t^2\theta^2m < 1$, which implies $t\theta < 1/\sqrt{m}$. We can conveniently set $k = 4\log_m (m/\eps)$ and do a union bound to say that with high probability $1 - o(\eps)$, no columns will have more than $k$ non-zero bits. Formally, 
    \begin{align}
      \label{eq:claim_high:column_max}
      \pr{\exists i \in [m], \wt(P_i) \geq 4\log_m (m/\eps) } \leq m (\dfrac{e}{\sqrt{m}})^{4\log_m (m/\eps)} = o (\eps).
    \end{align}
    Setting $k = O(\log_m(m/\eps))$ is essentially having $k$ being a constant when $\eps = 1/\poly(n)$.
    Next, we show that with high probability $1 - o(\eps)$, there are at most $h = 2\log(1/\eps)$ columns that have weight at least $2$. This is done again by using union bound over all possible $h$ subsets of columns,
    \begin{align}
      \label{eq:claim_high:union_column}
      \Pr \left[ \exists S \in \binom{[m]}{h} , \forall j \in S, \, \wt(P_j) \ge 2 \right] \leq {m \choose h} \Pr \left[\forall j \in [h], \, \wt(P_j) \ge 2 \right].
    \end{align}
    Since each row is a fixed weight vector, if a column happens to have a weight $\geq 2$, conditioned on this event, the probability of the next column having weight $\geq 2$ is smaller. This negative correlation property implies,
    \begin{align}
      \label{eq:claim_high:single_column}
      \Pr \left[\forall j \in [h], \, \wt(P_j) \ge 2 \right] \leq \Pr \left[\wt(P_1) \ge 2 \right]^h \leq (\dfrac{t^2\theta^2}{2})^{h}.
    \end{align}
    Therefore, combining \cref{eq:claim_high:union_column} and \cref{eq:claim_high:single_column}, gives us an upperbound
    \begin{align}
      \label{eq:claim_high:column_set}
      \Pr[\text{There exists at least } h \text{ columns with weight} \geq 2 ] \leq {m \choose h}(\dfrac{t^2\theta^2}{2})^{h} \leq (\dfrac{emt^2\theta^2}{2h})^h.
    \end{align}
    Since $mt^2\theta^2 < 1$. If $\eps \leq 1/2$, we can plug in $h = 2\log(1/\eps)$ to get 
    \[
        (\dfrac{emt^2\theta^2}{2h})^h \leq  (\dfrac{e}{4})^{2\log(1/\eps)} = o(\eps)
    \]
    Otherwise, if $\eps > 1/2$, we can plug in $h = 3$ to get
    \[
        (\dfrac{emt^2\theta^2}{2h})^h \leq  (\dfrac{e}{6})^3 = o(\eps)
    \]
    Therefore, if $h = \max(2\log(1/\eps), 3)$, the probability in \cref{eq:claim_high:column_set} is upper bounded by $o(\eps)$. 
    \\\\
    In summary, probability bounds we have on the two bad events are,
    \begin{itemize}[itemsep=0pt]
      \item  The probability that there exists a column in $R$ with more than $k=4 \log_m(m/\eps)$ is at most $o(\eps)$ (\cref{eq:claim_high:column_max}).
      \item  The probability that there exists
      $h = \max(2\log(1/\eps), 3)$ columns in $R$ with weight $\geq 2$ is at most $o(\eps)$ (\cref{eq:claim_high:column_set}).
    \end{itemize}
    By union bound, with probability $1 - o(\eps)$, neither of the bad events happens. In that scenario, the maximum size of $|C|$ is upper-bounded as follows,
    \[ |C| < k \cdot h = \max(8\log_m(m/\eps)\log(1/\eps), 12 \log_{m}(m/\eps)).\]
    Because the second value in the maximum expression happens only when $\eps > 1/2$, that also means $\log_{m}(m/\eps) \leq 1$. Therefore, the statement can be further simplified
    \[
        |C| < k \cdot h = \max(8\log_m(m/\eps)\log(1/\eps), 12)
    \]
    Hence, the probability of the abort in \cref{step:abort}, i.e $|C| \geq \max(8\log_m(m/\eps)\log(1/\eps), 12)$, is at most $o(\eps)$.
  \end{proof}
Now, assume that abort does not occur in \cref{step:abort}, we will show that in that case at least one iteration, the input $(RA, b^\prime)$ provided to $\Slpn$ in \cref{algo:query} has the same distribution as $(RA, RAs + \Ber(\tau)^t)$.

By definition of $C  = \{ j \in [t] \mid \exists i < j: \supp(R[i]) \cap \supp(R[j]) \neq \emptyset \}$, we get that for all distinct $p,q \in [t]\setminus C$,
\[
    \supp(R[p])\cap \supp(R[q]) = \emptyset.
\]
We can therefore split the smoothing matrix $R$ into two parts: `touched' submatrix $R_t \in \Field{2}^{|C|\times m}$, consisting of the rows indexed by the elements of $C$ and `untouched' submatrix $R_{ut} \in \Field{2}^{(t-|C|)\times m}$, consisting of the rows not indexed by the elements of $C$. Formally, viewing both $C$ and $[t]\setminus C$ as lists with arbitrary but fixed orderings, we define
\[
    R_t[i] := R[C[i]]
    \quad\text{for every } i \in [|C|],
\]
and
\[
    R_{ut}[i] := R[([t]\setminus C)[i]]
    \quad\text{for every } i \in [t-|C|].
\]
  
  To simplify the argument, ignore the ordering of the rows in $R_{ut}$ and $R_{t}$ as we can always combine them back to $R$ and maintain the original ordering.
  
   \medskip
  For the rows represented by $R_{ut}$, because the rows' supports are disjoint, each bit of $R_{ut} \cdot e$, where $e \gets \Ber(\eta)^{m}$, is independently and identically distributed. The distribution of $(R_{ut} e)_i$ is exactly the sum of $\theta \cdot m$ independent and identically distributed $\Ber(\eta)$ random variables. By \cref{fact:sum_bern}, this is $(R_{ut} e)_i \equiv \Ber(\tau)$ where $\tau = (1 - (1-2\eta)^{\theta \cdot m})/2$ and implies that  $(R_{ut}, R_{ut} e) \equiv (R_{ut}, \Ber(\tau)^{t - |C|})$. Then 
  \[
  (R_{ut}A, R_{ut}b) = (R_{ut}A, R_{ut}As + R_{ut}e) \equiv (R_{ut}A, R_{ut}As + \Ber(\tau)^{t - |C|}).
  \]

  For the rows represented by $R_{t}$, since we are enumerating all possible $\Field{2}^{|C|}$ assignments. There will be one assignment that is exactly the value of $v = R_{t}As$, and in this case the replacement procedure in \cref{algo:replace} sets the corresponding part of $b'$ to exactly $R_{t}As + \Ber(\tau)^{|C|}$.

  Finally, we combine $R_{t}$ and $R_{ut}$ back to $R$ respecting the original ordering. Assuming that we are in the iteration of getting a correct assignment, $v = R_{t}As$, the bits in $R_tb = R_tAs + R_te$ are replaced with $R_tAs + \Ber(\tau)^{|C|}$. This does not leak any information about $R_{ut} e$, and the query to $\Slpn$ is exactly the distribution $(RA, RAs + \Ber(\tau)^m)$. As $\Slpn$ has $\eps/2$ probability of success on this distribution and the probability of abort in \cref{step:abort} is at most $o(\eps)$, with probability at least $\eps/2-o(\eps)$, this execution $\Slpn(RA,b')$ will return $s' = s$. For the verification of the secret $s$ in \cref{algo:verify}, having the correct candidate secret $s^\prime = s$, the probability that error vector $b - As^\prime$ has Hamming weight greater than $2\eta m$ is exponentially small, and so the verification will fail with negligible probability. So the algorithm $\algS$ succeeds with probability at least $\eps/2-o(\eps)$.

  The remaining possibility is that the algorithm $\Slpn(RA,b')$ returns a candidate $s^\prime$ that is not equal to $s$ in some other iteration (i.e $v \neq R_{t}As$). This is still fine because $s^\prime$ must have the property that $\wt(b - As^\prime) \leq 2\eta m$ to pass the verification. Therefore, $S$  recovers a secret candidate $s^\prime$ from noisy codeword samples of $A$ with advantage at least $\eps/2 - o(\eps)$.

  The runtime of the algorithm is determined by the size of the matrix $A \in \Field{2}^{m \times n}$ and the number of queries called to $\Slpn$. Since $|C|$ is upper bounded by $\max(8\log_m(m/\eps)\log(1/\eps), 12)$, the number of iterations is $O(2^{\log_m(m/\eps)\log(1/\eps)}) = O((1/\eps)^{\log_m(m/\eps)})$ .Therefore, the runtime of $\algS$ is $\poly\left(n, m, T, (1/\eps)^{\log_m(m/\eps)}\right)$ 
\end{proofof}

\subsection{Small Advantage}
\label{sec:adv_small}
In this section, we prove \cref{lem:adv_small}, about the ability of an algorithm $\algD$ to distinguish random syndromes from uniform when $Adv(A) < \eps/2$ .

\lemmaadvsmall*

The algorithm $\algD$ that proves the lemma is described below, followed by the proof of the lemma. Here, $\Slpn$ is the assumed search algorithm for LPN that has running time $T$ and advantage $\eps$.

\begin{algorithm}[h!]
  \caption{Algorithm $\algD(A \in \Field{2}^{m \times n}, c^T \in \Field{2}^{1 \times n})$}
  \begin{spacing}{1.2}
    \begin{algorithmic}[1]
      \State Sample $i \gets [0, t - 1]$.                            \Comment{Select which hybrid to use}
      \State Sample matrices $R \gets \cW(m, \theta \cdot m)^{i}$, $B \gets \Field{2}^{(t - 1 - i) \times n}$
      \State Set $H \gets \begin{pmatrix}
        RA\\
        c^T\\
        B
      \end{pmatrix} \in \Field{2}^{t \times n}$                        \Comment{Create a sample from the hybrid distribution}
      \State Sample $s \gets \Field{2}^n$ and $e \gets \Ber(\tau)^t$  \Comment{Recall that $\tau =  (1 - (1-2\eta)^{\theta \cdot m})/2$}
      \If{$\Slpn(H, H s + e)$ returns $s$}                        \Comment{Check whether $\Slpn$ succeeds}
      \State \Return $1$
      \Else\ \Return $0$
      \EndIf
    \end{algorithmic}
  \end{spacing}
\end{algorithm}

\begin{proofof}{\cref{lem:adv_small}}
Fix an $A$ such that $Adv(A) < \eps/2$. We begin the proof by defining a sequence of hybrid distributions. For any $i \in [0, t]$, the hybrid distribution $\mathcal{H}_i$ over $\Field{2}^{t \times n}$ is defined by the following sampling process:
\begin{enumerate}[itemsep=0pt]
      \item Sample a matrix $R \gets \cW(m, \theta \cdot m)^{i}$.
      \item Sample a matrix $B \gets \Field{2}^{(t - i) \times n}$
    \item Output the row-concatenated matrix $\begin{pmatrix}
      RA\\
      B
    \end{pmatrix} \in \Field{2}^{t \times n}$.
\end{enumerate}

We now use the fact that, by our assumption in this case, the $\Slpn$ solver succeeds with low probability on the hybrid corresponding to $i=t$, whereas it succeeds with high probability on the hybrid corresponding to $i=0$. To be concrete, first
  fix any $i \in [0, t - 1]$, sample $R \gets \cW(m, \theta \cdot m)^{i}$, $B \gets \Field{2}^{(t - i - 1) \times n}$, we can make the following conclusion on the distribution $\cH$ of $H = \begin{pmatrix}
    RA\\
    c^T\\
    B
  \end{pmatrix}$.
  \begin{itemize}[itemsep=0pt]
    \item If $c \gets \Field{2}^m$, $\cH$ is exactly $\mathcal{H}_{i}$.
    \item If $c = r^TA$ where $r \gets \cW(m, \theta \cdot m)$, $\cH$ is exactly $\mathcal{H}_{i+1}$.
  \end{itemize}
  With $s$ and $e$ sampled as $s \gets \Field{2}^n$ and $e \gets \Ber(\tau)^t$, define $Adv_0, Adv_1$ as
  \[
  Adv_0 = \Pr_{\substack{ i\gets[0,t-1]\\ H\gets\cH_i}}[\Slpn(H, H s + e) = s],\quad Adv_1 = \Pr_{\substack{ i\gets[0,t-1]\\ H\gets\cH_{i+1}}}[\Slpn(H, Hs + e) = s].
  \]
  Then, the behavior of the algorithm $\algD$ we defined can be described exactly with $Adv_0, Adv_1$
  \[
  Adv_0 = \Pr_{\substack{c \gets \Field{2}^n}}[\algD(A, c^T) = 1],\quad  Adv_1 = \Pr_{r \gets \cW(m, \theta \cdot m)}[\algD(A, r^TA) = 1].
  \]
  By a standard hybrid argument, $Adv_0$ and $Adv_1$ have a gap of $Adv_0 - Adv_1 \geq \eps/2t$. This is stated in the following claim, whose proof we defer to \cref{proof:hybrid_argument}.
  \begin{restatable}{claim}{hybridargument}\emph{(Hybrid Argument).}  \label{claim:hybrid_argument}
    $
    Adv_0 - Adv_1 \geq \eps/2t.
    $
  \end{restatable}
  Therefore, algorithm $\algD$ can distinguish random syndromes of $A$ from uniform with advantage at least $\eps/2t$. As $\algD$ only calls $\Slpn$ once, its running time is $\poly(n,m) + T$.
\end{proofof}

\section{Corollaries}
\label{sec:corollaries}

In this section, we discuss some implications of \cref{thm:main}. We show how the theorem gives meaningful hardness for LPN with parameters that are within the statistical-computational gap. In particular, we will show that \cref{thm:main} gives partial worst-case hardness for some parameters for LPN that can be used to construct public-key cryptography. The next corollary describes the contrapositive statement of \cref{thm:main} with parameters simplified and rearranged.

\begin{corollary}[Hardness for LPN]
    \label{cor:lpnmain}
    Consider any polynomial $t(n) \in \mathbb{Z}$ with $t(n) > n$, constants $\alpha, \beta \in (0.01, 0.99)$ with $\alpha + \beta \in (0.01, 0.99)$, and $m(n) = \ceil{2t^2n^{2\beta}}$. Suppose that for every pair of PPT algorithms $(\algS, \algD)$, there exists a negligible function $negl$ such that, for all large enough $n$, there exists an $\F_2$-matrix $A$ of dimension $m \times n$ for which both of these conditions are true: 
    \begin{itemize}[itemsep=0pt]
        \item $\Pr[\algS(A, As+e) = s^\prime\  \mathrm{and}\  \wt(As^\prime - (As + e)) \leq 2 m/n^{\alpha+\beta}]  \leq negl(n)$, where $s \gets \Field{2}^n$, $e \gets Ber(n^{-(\alpha + \beta)})^{m}$.
        \item $\abs{\Pr[\algD(A, c^T) = 1] - \Pr[\algD(A, r^TA) = 1]} \leq negl(n)$, where $c \gets \Field{2}^n, r \gets \cW(m,n^{\beta})$.
    \end{itemize}    
    Then, for the Search LPN problem with sample count $t(n)$, secret dimension $n$, and noise parameter $n^{-\alpha}$, no PPT algorithm has non-negligible advantage.
\end{corollary}
\begin{proof}
  This corollary can be proven using the contrapositive statement of \cref{thm:main} with an appropriate setting of variables as follows:
  \begin{align*}
    \eta = \frac{1}{n^{\alpha+\beta}} \quad \theta = \frac{n^\beta}{m}
  \end{align*}
  The resulting noise rate of LPN from the theorem can be bounded as follows using the Bernoulli inequality,
    \[
        \tau = \frac{1}{2}\left(1-\left(1-\frac{2}{n^{\alpha+\beta}}\right)^{n^{\beta}}\right) \leq n^{-\alpha}.
    \]
    The noise rate in the LPN instance can later be artificially increased to be equal to $n^{-\alpha}$ as required. Next, we can also check that the condition of $t^2\theta^2m < 1$ required by \cref{thm:main} (that the probability of intersection between two rows is bounded) is also satisfied ,
    \[
    \dfrac{(t)^2(n^{\beta})^2}{m}  \leq \dfrac{(t^2)n^{2\beta}}{2t^2n^{2\beta}} = \dfrac{1}{2} < 1.
    \] 
    So if there is a polynomial-time algorithm for Search LPN that has non-negligible advantage $\eps(n)$ on infinitely many values of $n$, then by \cref{thm:main}, it contradicts the assumption in the statement of the corollary. Note that the runtime of $\algS$ resulting from the theorem is $\poly\left(n, m, T, (1/\eps)^{\log_m(m/\eps)}\right)$, and when $\eps = 1/\poly(n)$, we have that  $(1/\eps)^{\log_m(m/\eps)} = \poly(n)$, so the runtime is still polynomial in $n$. That concludes the proof of the corollary.
\end{proof}


\paragraph{Non-Triviality of Hardness.} For the reductions in \cref{thm:main} and \cref{cor:lpnmain} to be non-trivial and useful for cryptography, they have to show that under a reasonable worst-case hardness assumption, we can infer the hardness of LPN with parameters for which it is possible to solve statistically.

For the Search LPN problem to be statistically feasible to solve, the parameters should be in a regime where the instance $(A,As+e)$ contains enough information about the secret $s$. Using the fact that random linear codes achieve Shannon capacity for the binary symmetric channel~\cite{guruswami2012essential}, an LPN instance with secret size $n$, noise rate $\eta$, and sample size $t \gg n/(1-h(\eta))$, where $h$ is the binary entropy function, satisfies this condition. For $\eta = n^{-\alpha}$ with $\alpha \in (0.01, 0.99)$, this condition is satisfied if $t > n + \omega(n^{1-\alpha} \log{n})$.

To the best of our knowledge, to design a pair of algorithms $(\algS,\algD)$ such that $\algS$ can decode the code generated by a matrix $A$ from random noise while $\algD$ can distinguish random syndromes generated by it from uniform is to pick the optimal algorithm for each task separately, just keep the more efficient one as $\algS$ or $\algD$, and let the other algorithm simply not do anything. If $\alpha, \beta \in (0.01,0.99)$, the best such algorithms run in time $\exponential(O(n^{1-(\alpha + \beta)}))$ for $\algS$ and $\exponential(O(n^{\beta}))$ for $\algD$~\cite{Prange62}.
Thus, if $(\alpha+\beta)\in(0.01,0.99)$, the best algorithms for the worst-case problem are still sub-exponential time.





\paragraph{Parameter settings.} It remains to show that we can set our parameters so that the above conditions are simultaneously satisfied. Some of these are as follow. For each setting of parameters below, we make the corresponding assumption as stated in the hypothesis of \cref{cor:lpnmain}, and the conclusion is the result of applying the corollary with these parameters.

\begin{enumerate}
  \item Setting $\alpha = 0.1$, $\beta = 0.1$, and $m = \ceil{8n^{2.2}}$ implies the hardness of LPN with $t = 2n$ samples and noise rate $n^{-0.1}$.
  \item Setting $\alpha = 0.5$, $\beta = 0.45$, and $m = \ceil{8n^{2.9}}$ implies the hardness of LPN with $t = 2n$ samples and noise rate $n^{-0.5}$.
  \item Setting $\alpha = 0.8$, $\beta = 0.1$, and $m = \ceil{2n^{6.2}}$ implies the hardness of LPN with $t = n^3$ samples and noise rate $n^{-0.8}$.
\end{enumerate}

Note that the last two parameter settings above, together with known search-to-decision reductions for LPN~\cite{JACM:Regev09,C:AppIshKus07,C:MicMol11}, are sufficient for use in Alekhnovich's construction of a Public-Key Encryption (PKE) scheme from the hardness of LPN~\cite{FOCS:Alekhnovich03}. There are other viable settings of the parameter $\beta$ for which this would be true as well. The following is one specific corollary of the many that are correspondingly possible.

\begin{corollary}
  \label{cor:pke}
  If the hypothesis of \cref{cor:lpnmain} is satisfied for the parameter setting $t = 2n$, $\alpha = 0.5$, $\beta = 0.45$ and $m = \ceil{8n^{2.9}}$, then there is a secure construction of PKE.
\end{corollary}

\section{Acknowledgements}

We would like to thank Yuval Ishai for pointing us to relevant references on barriers to the statistical smoothing approach.

This work was supported by the National Research Foundation, Singapore, under awards no. NRF-NRFF14-2022-0010 and NRF-NRFI09-0005.

\bibliographystyle{alphaurl}
\bibliography{abbrev0,crypto,local}

\appendix

\section{Deferred Proof}

\pilinguplemma*

\begin{proof}
    \label{proof:piling_up_lemma}
    Let $Y_i := (-1)^{X_i}\in\{\pm1\}$. Then $\exp{Y_i} = (1-\eta)\cdot 1 + \eta\cdot(-1) = 1-2\eta$. Notice that,
\[
(-1)^{S} = Y_1 \cdot Y_2 \cdots Y_k \;.
\]
Since $Y_1, \ldots, Y_k$ are independently distributed, we have that 
\[
\exp{(-1)^S} \;=\; \exp{\prod_{i=1}^k (-1)^{X_i}}
\;=\; \prod_{i=1}^k \exp{(-1)^{X_i}}
\;=\; (1-2\eta)^k.
\]
From $\exp{(-1)^S} = \Pr[S=0]-\Pr[S=1] = 1 - 2\,\Pr[S=1]$.  
Hence 
\[
\Pr[S=1] \;=\; \frac{1-(1-2\eta)^k}{2}. 
\]
\end{proof}

\hybridargument*

\begin{proof}
    \label{proof:hybrid_argument}
    Suppose $s \gets \Field{2}^{n}, e \gets \mathrm{Ber}(\tau)^t$ throughout the proof. Recall that,
    \[
        Adv = 
        \Pr_{\substack{R \gets \cW(m, \theta \cdot m)^{t}}}[\Slpn(RA, RAs + e) = s].
    \]
    From $\Slpn$ having advantage $\eps$ and $Adv < \eps/2$, 
    \begin{equation}
        \label{eq:hybrid_advantage}
        \Pr_{\substack{B \gets \Field{2}^{t \times n}}}[\Slpn(B, Bs+e) = s] - \Pr_{\substack{\\ R \gets \cW(m, \theta \cdot m)^{t}}}[\Slpn(RA, RAs + e) = s] \geq  \eps/2.
    \end{equation}
    Note that $\mathcal{H}_{t}$ is exactly the distribution $RA$, while $\mathcal{H}_{0}$ is the uniform distribution.  Telescoping,
    \begin{align*}
        &\Pr_{\substack{B \gets \Field{2}^{t \times n}}}[\Slpn(B, Bs+e) = s] - \Pr_{\substack{\\ R \gets \cW(m, \theta \cdot m)^{t}}}[\Slpn(RA, RAs + e) = s]\\
        = &\sum_{i \in [0, t-1]} \Pr_{\substack{H_{i} \gets \mathcal{H}_{i}}}[\Slpn(H_{i}, H_{i}s + e) = s] - \Pr_{\substack{H_{i+1} \gets \mathcal{H}_{i+1}}}[\Slpn(H_{i+1}, H_{i+1}s + e) = s].
    \end{align*}
    There are two components in the equation, look at each component individually, by the law of total probability,
    \begin{align*}
        &\sum_{i \in [0, t-1]} \Pr_{\substack{H_{i} \gets \mathcal{H}_{i}}}[\Slpn(H_{i}, H_{i}s + e) = s] \\
        = &\ t \sum_{i \in [0, t-1]} \dfrac{1}{t}  \cdot \Pr_{\substack{H_{i} \gets \mathcal{H}_{i}}}[\Slpn(H_{i}, H_{i}s + e) = s]\\
        = &\ t \sum_{i \in [0, t-1]} \Pr_{X \gets [0, t - 1]}[X = i] \Pr_{\substack{H_{X} \gets \mathcal{H}_{X}}}[\Slpn(H_{X}, H_{X}s + e) = s\ |\ X = i]\\
        = &\ t \Pr_{\substack{i \gets [0, t-1]\\ H_{i} \gets \mathcal{H}_{i}}}[\Slpn(H_{i}, H_{i}s + e) = r] = t\cdot Adv_0.
    \end{align*}
    By a similar argument, 
    \[
        \sum_{i \in [0, t-1]} \Pr_{\substack{H_{i+1} \gets \mathcal{H}_{i+1}}}[\Slpn(H_{i+1}, H_{i+1}s + e) = s] = t\cdot Adv_1.
    \]
    Substitute into \cref{eq:hybrid_advantage}, and that concludes the claim
    \[
        t \cdot Adv_0 - t \cdot Adv_1 \geq \eps/2 \implies Adv_0 - Adv_1 \geq \eps/2t.
    \]
\end{proof}

\end{document}